\documentclass[conference,10pt]{IEEEtran}
\IEEEoverridecommandlockouts

\usepackage{cite}
\usepackage{graphicx}
\usepackage{array}
\usepackage{mdwmath}
\usepackage{amssymb}
\usepackage{mdwtab}
\usepackage{stfloats}
\usepackage{amsmath,amsthm}
\usepackage{threeparttable}
\usepackage{color}
\usepackage{bm}
\usepackage{url}
\usepackage{algpseudocode}
\usepackage{algorithm}
\usepackage{multicol}
\usepackage{epstopdf}
\usepackage{flushend}
\usepackage[caption=false,font=footnotesize]{subfig}

\algrenewcommand{\algorithmicrequire}{\textbf{Input:}}
\algrenewcommand{\algorithmicensure}{\textbf{Output:}}

\newtheorem{theorem}{Theorem}
\newtheorem{remark}{Remark}
\newtheorem{lemma}{Lemma}

\begin{document}

\title{UAV-Enabled Short-Packet Communication via Fluid Antenna Systems}

\author{
\IEEEauthorblockN{
Xusheng Zhu\IEEEauthorrefmark{1}, 
Kai-Kit Wong\IEEEauthorrefmark{1}\IEEEauthorrefmark{2}, 
Hanjiang Hong\IEEEauthorrefmark{1}, 
Han Xiao\IEEEauthorrefmark{3}, 
Hao Xu\IEEEauthorrefmark{4}, 
Tuo Wu\IEEEauthorrefmark{5}, 
and 
Chan-Byoung Chae\IEEEauthorrefmark{2}
}

\IEEEauthorblockA{\IEEEauthorrefmark{1}Department of Electronic and Electrical Engineering, University College London, London, UK}
\IEEEauthorblockA{\IEEEauthorrefmark{2}Yonsei Frontier Lab \& School of Integrated Technology, Yonsei University, Seoul, Korea}
\IEEEauthorblockA{\IEEEauthorrefmark{3}School of Information and Communications
Engineering, Xi’an Jiaotong University, Xi’an, China.}
\IEEEauthorblockA{\IEEEauthorrefmark{4}National Mobile Communications Research Laboratory, Southeast University, Nanjing, China}
\IEEEauthorblockA{\IEEEauthorrefmark{5}Department of Electronic Engineering, City University of Hong Kong, Hong Kong
}
\thanks{This work was supported in part by the EPSRC under Grant EP/W026813/1; in part by NSFC Grant 624B2094; in part by the Outstanding Doctoral Graduates Development Scholarship of Shanghai Jiao Tong University.}
}

\maketitle

\begin{abstract}
This paper develops a framework for analyzing UAV-enabled short-packet communication, leveraging fluid antenna system (FAS)-assisted relaying networks. Operating in the short-packet regime and focusing on challenging urban environments, we derive novel, closed-form expressions for the block error rate (BLER). This is achieved by modeling the spatially correlated Nakagami-$m$ fading link via a tractable eigenvalue-based approach. A high-signal-to-noise ratio (SNR) asymptotic analysis is also presented, revealing the system's fundamental diversity order. Building on this analysis, we formulate a novel energy efficiency (EE) maximization problem that, unlike idealized models, uniquely incorporates the non-trivial time and energy overheads of FAS port selection. An efficient hierarchical algorithm is proposed to jointly optimize key system parameters. Numerical results validate our analysis, demonstrating that while FAS provides substantial power gains, the operational overhead creates a critical trade-off. This trade-off dictates an optimal number of FAS ports and a non-trivial optimal UAV deployment altitude, governed by the balance between blockage and path loss. This work provides key insights for FAS-aided UAV communications.
\end{abstract}


\section{Introduction}
The evolution toward sixth-generation (6G) wireless networks is driven by the imperative to support ultra-reliable low-latency communications (URLLC) for mission-critical services~\cite{saad2019vision}. To meet stringent latency requirements, URLLC relies on short-packet communication, operating in the short-packet regime~\cite{she2017radio}. This renders classical Shannon capacity theorems inadequate, as performance is now governed by the trade-off between coding rate, blocklength, and a non-negligible error probability~\cite{polyanskiy2010channel}. Concurrently, unmanned aerial vehicles (UAVs) have emerged as a pivotal technology for enhancing 6G reliability and coverage~\cite{wu2021overview}. However, a key challenge persists in providing reliable links to ground users with strict size, weight, and power (SWaP) constraints, where achieving spatial diversity with conventional multiple-antenna technology is often impractical~\cite{yuan2022per}.

To this end, the fluid antenna system (FAS) has been introduced as a transformative paradigm capable of harvesting spatial diversity within a compact form factor~\cite{new2025a}. Compared to traditional fixed-position antenna (FPA), FAS exploits spatial diversity using a single radio frequency (RF) chain by reconfiguring the antenna position within a specified region. Seminal work in \cite{wong2021fluid} established its theoretical foundations. Subsequent research has expanded its applications in diverse domains, including physical layer security~\cite{wong2020pel}, multiple access schemes~\cite{wong2022fama}, and massive connectivity~\cite{wong2024compact}. Consequently, a body of research has begun to explore the integration of FAS with UAVs, investigating its use in various network architectures and for interference management~\cite{cha2025uav, li2025radi, xu2025fluid}.

Despite these advancements, the application of FAS to URLLC-enabled UAV missions remains nascent. A critical limitation of prior FAS-UAV works~\cite{cha2025uav, li2025radi, xu2025fluid} is their reliance on the infinite blocklength assumption, using classical metrics like ergodic capacity that are fundamentally misaligned with the short-packet nature of URLLC. 
Against this background, this paper is motivated by a convergence of critical research gaps: 1) the lack of a unified analytical framework for FAS-UAV systems under short-packet theory; 2) the reliance on oversimplified channel models; and 3) the universal neglect of practical FAS operational overheads, which are non-trivial in latency-sensitive and power-constrained UAVs. To address these limitations, this work establishes a holistic framework for the joint analysis and optimization of FAS-enabled UAV relaying systems deployed in urban scenarios.
The main contributions are:
\begin{itemize}
\item We derive novel closed-form block error rate (BLER) expressions for FAS-UAV systems in the FBL regime, capturing correlated Nakagami-$m$ fading for urban scenarios using a tractable eigenvalue-based model.
\item We derive high-signal-to-noise ratio (SNR) asymptotic BLER expressions, quantifying the system's diversity order and proving the existence of a first-hop-bottlenecked error floor.
\item We formulate a novel energy efficiency (EE) maximization problem that integrates the practical time and energy overheads of FAS port selection, and propose an efficient hierarchical algorithm to solve the resulting non-convex mixed-integer problem.
\item We reveal critical design insights, including the existence of an optimal finite number of FAS ports and an optimal UAV deployment altitude governed by the balance between blockage and path loss.
\end{itemize}

\section{System Model}\label{sec:system_model}
We consider a three-node UAV-aided downlink relaying system, as illustrated in Fig.~\ref{fig:frame}. The system comprises a base station (BS), a half-duplex decode-and-forward (DF) UAV relay, and a FAS-equipped user equipment (UE) operating in a challenging urban environment.

\subsection{System Geometry and Probabilistic Path Loss}
We consider a 3-D Cartesian coordinate system where the BS and UE are at fixed locations, and the UAV flies on a circular trajectory of radius $r$ at altitude $Z_{\rm U}$. The slant ranges for the BS–UAV and UAV–UE links, $d_1(\theta)$ and $d_2(\theta)$, are respectively given by
\begin{small}
\begin{align}
d_1(\theta)&\!=\!\sqrt{(r\cos\theta\!-\!X_{\rm B})^2 \!+\! (r\sin\theta\!-\!Y_{\rm B})^2 \!+\! (Z_{\rm U}\!-\!Z_{\rm B})^2},\\
d_2(\theta)&\!=\!\sqrt{(r\cos\theta\!-\!X_{\rm E})^2 \!+\! (r\sin\theta\!-\!Y_{\rm E})^2 \!+\! (Z_{\rm U}\!-\!Z_{\rm E})^2}.
\end{align}
\end{small}

\begin{figure}[t]
 $\,$ \centering
 $\,$ \includegraphics[width=0.9\columnwidth]{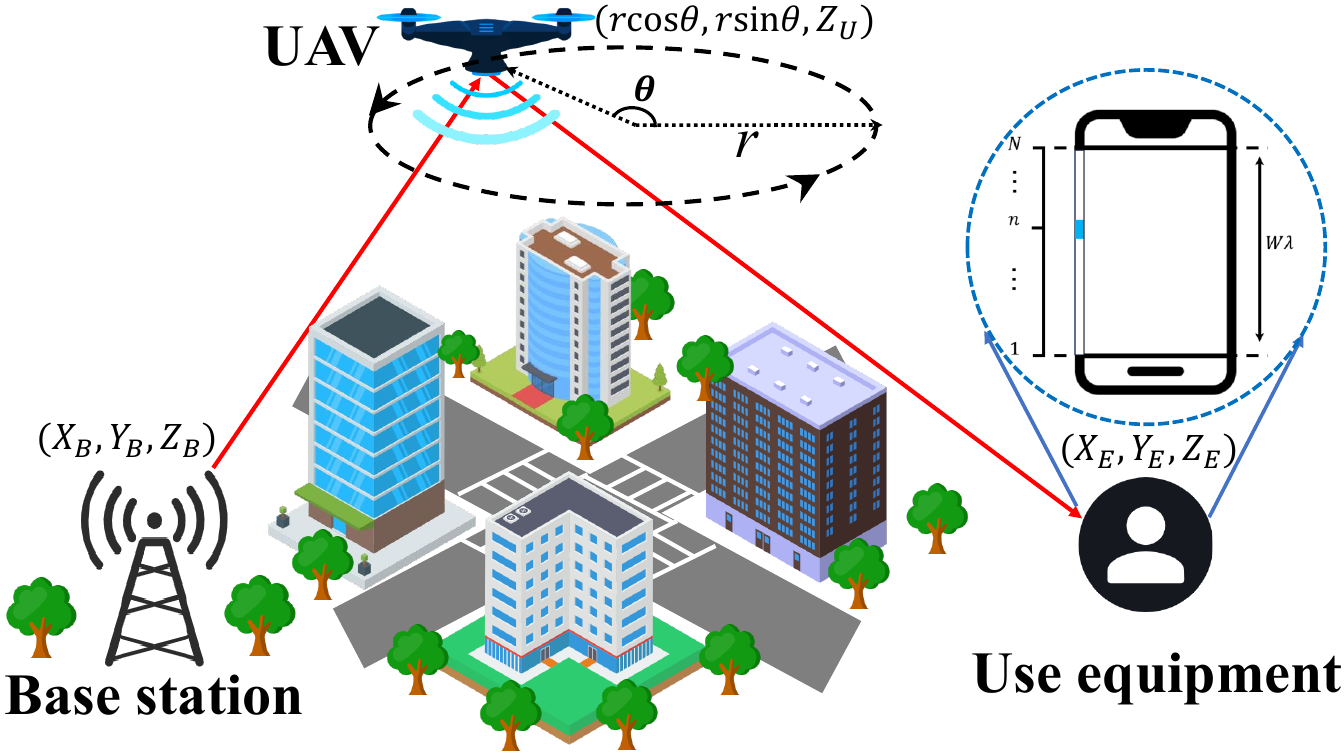} 
 $\,$ \caption{System model.}
 $\,$ \label{fig:frame}
\end{figure}
The urban scenario incorporates a probabilistic mixture of line-of-sight (LoS) and non-line-of-sight (NLoS) conditions, indexed by $k \in \{\rm LoS, NLoS\}$. The probability of an LoS condition for hop $i$ is modeled as a function of the elevation angle $\phi_i(\theta)$ as \cite{hour2014opt}
\begin{equation}
P^{\rm LoS}_{i}(\theta)=\frac{1}{1+a\exp\left(-b[\phi_i(\theta)-a]\right)},
\end{equation}
with $P^{\rm NLoS}_{i}(\theta)=1-P^{\rm LoS}_{i}(\theta)$. The elevation angle is $\phi_i(\theta)=\frac{180}{\pi}\arcsin(\Delta Z_i/d_i(\theta))$, and for a dense urban environment, $a=12.08$ and $b=0.11$.
The large-scale fading coefficient for a link of type $k$ is $\beta_{i}^{k} = (c/4\pi f_c d_i(\theta))^2 10^{-\eta_{i}^{k}/10}$, where $\eta_{i}^{k}$ is the excess path loss.

\subsection{Channel and Signal-to-Noise Ratio (SNR) Model}
We analyze the instantaneous SNR for each hop, which forms the basis of our performance analysis.

\subsubsection{First Hop (BS-to-UAV)}
The first-hop channel $g$ is subject to Nakagami-$m$ fading, where its power $|g|^2$ follows a Gamma distribution with parameter $m_1$. The instantaneous SNR at the UAV for link type $k$ is
\begin{equation}
\gamma_1^k = \frac{P_1 \beta_1^k(\theta)}{\sigma^2} |g^k|^2 \triangleq \bar{\gamma}_1^k(\theta) |g^k|^2,
\end{equation}
where $P_1$ is the BS transmit power, $\sigma^2$ is the additive white Gaussian noise (AWGN) power, and $\bar{\gamma}_1^k(\theta)$ is the average SNR.

\subsubsection{Second Hop (UAV-to-UE)}
The second-hop channel $\mathbf{h} \in \mathbb{C}^{N \times 1}$ terminates at the $N$-port FAS. It is spatially correlated, modeled as $\mathbf{h} = \mathbf{U}\mathbf{\Lambda}^{1/2}\mathbf{g}$, where $\mathbf{J} = \mathbf{U}\mathbf{\Lambda}\mathbf{U}^H$ is the Jakes' correlation matrix $J_{m,n} = J_0( 2\pi W(m-n)/(N-1) )$. The vector $\mathbf{g}$ contains i.i.d. components, where each $|g_n|^2$ is Gamma distributed with parameter $m_2^k$.

A direct analysis of the selected gain $\max\{|h_n|^2\}$ is mathematically intractable. To ensure tractability, we adopt a widely-used model based on the $N_{\text{eff}} \triangleq \text{rank}(\mathbf{J})$ effective independent diversity branches \cite{zhu2025fas}. The analytical model for the selected channel power gain is defined as
\begin{equation}\label{eq:analytical_model_def}
 |h_{\rm FAS}|^2 \triangleq \max\{\lambda_1|g_1|^2, \dots, \lambda_{N_{\mathrm{eff}}}|g_{N_{\mathrm{eff}}}|^2\}.
\end{equation}
The instantaneous SNR at the UE, given transmit power $P_2$, is
\begin{equation}
\gamma_{2}^{k} = \frac{P_2 \beta_2^k(\theta)}{\sigma^2} |h_{\rm FAS}^{k}|^2.
\end{equation}
For analytical convenience in the following sections, we normalize this expression as $\gamma_{2}^{k} = \overline{\gamma}_{2}^{k}(\theta) \frac{|h_{\rm FAS}^{k}|^{2}}{\sum_{n=1}^{N_{\rm eff}}\lambda_{n}}$, where the average SNR is defined as $\overline{\gamma}_{2}^{k}(\theta) = \frac{P_2 \beta_2^k(\theta) \sum_{n=1}^{N_{\rm eff}} \lambda_n}{\sigma^2}$.

\section{Performance Analysis}
For a given finite block length $L$ and packet size $B$, the coding rate $R=B/L$ is approximated as \cite{polyanskiy2010channel}
\begin{equation}\label{Rnx}
R(L,\epsilon) \approx C(\gamma)-\sqrt{\frac{Z(\gamma)}{L}}Q^{-1}(\epsilon),
\end{equation}
where $C(\gamma)=\log_2(1+\gamma)$ is the Shannon capacity, $Z(\gamma)=\left(1-\frac{1}{(1+\gamma)^2}\right)(\log_2e)^2$ is the channel dispersion, and $\epsilon$ is the error probability.

\subsection{Block Error Rate (BLER) Analysis}
Under the urban scenario, the instantaneous BLER $\epsilon_i^k$ for link type $k \in \{\mathrm{LoS}, \mathrm{NLoS}\}$ is
\begin{equation}
\epsilon_i^k\approx Q\left(\frac{C(\gamma^k_i)-R}{\sqrt{Z(\gamma^k_i)/L}}\right).
\end{equation}
The average BLER $\bar{\epsilon}_i^k(\theta)$ is the expectation of $\epsilon_i^k$. To render this integral tractable, we employ a highly accurate piecewise linear approximation \cite{yuan2022per} within a defined SNR range $[\rho_L, \rho_H]$ as
\begin{equation}
Q\left(\frac{C(\gamma)-R}{\sqrt{Z(\gamma)/L}}\right)
\approx
\begin{cases} 
1, & \gamma \leq \rho_L \\
\frac{1}{2} - \chi(\gamma - \tau), & \rho_L < \gamma < \rho_H \\
0, & \gamma \geq \rho_H 
\end{cases}
\end{equation}
where $\chi = 1/\sqrt{2\pi(2^{R}-1)/L}$, $\tau=2^{R}-1$, $\rho_L = \tau-1/(2\chi)$, and $\rho_H=\tau+1/(2\chi)$.
This simplifies the average BLER integral to
\begin{equation}\label{epaprir_urban}
\bar{\epsilon}^{k}_i(\theta)\approx
\chi\int_{\rho_L}^{\rho_H}
F_{\gamma_i^{k}}(x)dx,
\end{equation}
where $F_{\gamma_i^{k}}(x)$ is the cumulative distribution function (CDF) of $\gamma_i^k$.

\begin{lemma}\label{lema1_urban}
The first-hop SNR CDF for link type $k$ is
\begin{equation}
F_{\gamma_{1}^{k}}(x)=1-e^{-x\vartheta_{1}^{k}}\sum\nolimits_{j=0}^{m_1-1}\frac{(x\vartheta_{1}^{k})^{j}}{j!},
\end{equation}
where $\vartheta_1^k(\theta) \triangleq m_1 / \bar{\gamma}_1^k(\theta)$.
\end{lemma}
\emph{Proof:}
Please see Appendix \ref{plema1_urban_proof}. \hfill$\blacksquare$

\begin{lemma}\label{lema2_urban}
The second-hop FAS SNR CDF for link type $k$ is
\begin{equation}\label{ga2kxl4}
F_{\gamma_{2}^{k}}(x)=\prod_{n=1}^{{N_{\rm eff}} }\left(1-e^{-\frac{x\vartheta_{2}^k}{\lambda_{n}}}\sum_{j=0}^{m_2^k-1}\frac{1}{j!}\left(\frac{x\vartheta_{2}^k}{\lambda_{n}}\right)^{j}\right),
\end{equation}
where $m_2^k$ is the Nakagami-m parameter and $\vartheta_{2}^{k}(\theta) \triangleq m_2^k \sum_{n=1}^N \lambda_n / \bar{\gamma}_2^k(\theta)$.
\end{lemma}
\emph{Proof:} Please see Appendix \ref{plema2_urban_proof}. \hfill$\blacksquare$

\subsubsection{First-Hop Average BLER}
Substituting Lemma~\ref{lema1_urban} into \eqref{epaprir_urban}, the closed-form BLER component is
\begin{equation}\label{eq:urban_hop1_k}
\begin{aligned}
\overline{\epsilon}_{1}^{k}(\theta) \approx \chi \bigg[& (\rho_H - \rho_L) - \frac{1}{\vartheta_1^k} \sum_{j=0}^{m_1-1} \Big( e^{-\rho_L\vartheta_1^k}\sum_{l=0}^{j}\frac{(\rho_L\vartheta_1^k)^l}{l!} \\
&\hspace{1.5cm}- e^{-\rho_H\vartheta_1^k}\sum_{l=0}^{j}\frac{(\rho_H\vartheta_1^k)^l}{l!} \Big)\bigg].
\end{aligned}
\end{equation}
The total average BLER for the first hop is the expectation over the link probabilities as
\begin{equation}
\bar{\epsilon}_{1}^{\mathrm{US}}(\theta) = \bar{\epsilon}_{1}^{\mathrm{LoS}}(\theta)P_{1}^{\mathrm{LoS}}(\theta) + \bar{\epsilon}_{1}^{\mathrm{NLoS}}(\theta)P_{1}^{\mathrm{NLoS}}(\theta). \label{eq:urban_hop1_final}
\end{equation}

\subsubsection{Second-Hop Average BLER}
Similarly, substituting Lemma~\ref{lema2_urban} into \eqref{epaprir_urban} and solving via inclusion-exclusion yields
\begin{align}
\overline{\epsilon}_{2}^{k}(\theta) &\approx{} \chi(\rho_H - \rho_L) + \chi \sum_{\substack{\mathcal{S} \subseteq \{1, \dots, N_{\rm eff}\} \\ \mathcal{S} \neq \emptyset}} (-1)^{|\mathcal{S}|} \times \nonumber \\
& \sum_{a=0}^{|\mathcal{S}|(m_2^k-1)} c_a^k(\mathcal{S}) \left[ \mathcal{G}(\rho_H; a, b_{\mathcal{S}}^k) - \mathcal{G}(\rho_L; a, b_{\mathcal{S}}^k) \right], \label{eq:urban_hop2_k}
\end{align}
where $c_a^k(\mathcal{S})$, $b_{\mathcal{S}}^k$, and $\mathcal{G}(\cdot)$ are derived from the expansion of \eqref{ga2kxl4}.

\begin{theorem}\label{theo1_urban}
In the high-SNR regime, the average BLER of the second hop for link type $k$ is approximated by
\begin{equation}
\begin{aligned}
\bar{\epsilon}_2^k \approx& \frac{\chi \left( \rho_H^{m_2^k {N_{\rm eff}} + 1} - \rho_L^{m_2^k {N_{\rm eff}} + 1} \right)}{m_2^k {N_{\rm eff}} + 1} \\
&\times\left(\frac{(\vartheta_2^k)^{m_2^k}}{\Gamma(m_2^k+1)}\right)^{{N_{\rm eff}} } \prod_{n=1}^{{N_{\rm eff}} } \lambda_n^{-m_2^k}. 
\end{aligned}
\end{equation}
\end{theorem}
\emph{Proof:} Please see Appendix \ref{proof_theo1_urban}. \hfill$\blacksquare$

\begin{remark}
Theorem \ref{theo1_urban} reveals that the diversity order of the FAS-enabled link is $m_2^k N_{\rm eff}$, jointly determined by the fading severity $m_2^k$ and the spatial degrees of freedom $N_{\rm eff}$.
\end{remark}

The total average BLER for the second hop is then
\begin{equation}
\bar{\epsilon}_{2}^{\mathrm{US}}(\theta) = \bar{\epsilon}_{2}^{\mathrm{LoS}}(\theta)P_{2}^{\mathrm{LoS}}(\theta) + \bar{\epsilon}_{2}^{\mathrm{NLoS}}(\theta)P_{2}^{\mathrm{NLoS}}(\theta). \label{eq:urban_hop2_final}
\end{equation}

\subsubsection{End-to-End BLER and Asymptotic Analysis}
The end-to-end BLER, $\bar{\epsilon}_{T}^{\mathrm{US}}(\theta)$, is obtained by combining the hop-level results as
\begin{equation}
\overline{\epsilon}_{T}^{\mathrm{US}}(\theta) 
= 1 - (1-\overline{\epsilon}_{1}^{\mathrm{US}}(\theta))(1-\overline{\epsilon}_{2}^{\mathrm{US}}(\theta)).
\end{equation}
The overall average BLER, $\bar{\epsilon}_{O}^{\mathrm{US}}$, is found by averaging over the UAV's trajectory as
\begin{equation}
\overline{\epsilon}_{O}^{\mathrm{US}} 
= \frac{1}{2\pi} \int_{0}^{2\pi} \overline{\epsilon}_{T}^{\mathrm{US}}(\theta) d\theta. 
\end{equation}
As this integral is intractable, we employ $M$-point Gauss-Chebyshev Quadrature (GCQ) as
\begin{equation}
\bar{\epsilon}_{O}^{\mathrm{US}}\;\approx\;\frac{\pi}{M}\sum_{m=1}^{M} w_m\, \bar{\epsilon}_{T}^{\mathrm{US}}(\theta_m),
\end{equation}
where the nodes $\theta_m$ and weights $w_m$ are given by
\begin{equation}\label{overbm_urban}
\begin{aligned}
 \theta_m = \pi \,x_m + \pi, \quad w_m = \frac{\pi}{M} \, \frac{1}{\sqrt{1-x_m^2}}, 
\end{aligned}
\end{equation}
with $x_m$ as the Chebyshev roots on $[-1, 1]$.

\begin{theorem}\label{theo2_urban_floor}
The end-to-end performance is limited by an error floor. For a fixed $P_1$ and $P_2 \to \infty$, the overall BLER converges to a floor determined solely by the first hop, given by
\begin{small}
\begin{align}
\lim_{P_2 \to \infty} \overline{\epsilon}_{O}^{\mathrm{US}} = \frac{1}{2\pi} \int_{0}^{2\pi} \overline{\epsilon}_{1}^{\mathrm{US}}(\theta) d\theta. \label{eq:corrected_floor} 
\end{align}
\end{small}%
\end{theorem}

\emph{Proof:}
As $P_2 \to \infty$, the second hop becomes error-free, $\overline{\epsilon}_{2}^{\mathrm{US}}(\theta) \to 0$. The end-to-end BLER $\overline{\epsilon}_{T}^{\mathrm{US}}(\theta)$ thus converges to $1 - (1 - \overline{\epsilon}_{1}^{\mathrm{US}}(\theta))(1 - 0) = \overline{\epsilon}_{1}^{\mathrm{US}}(\theta)$. Averaging over $\theta$ completes the proof.\hfill$\blacksquare$

\begin{remark}
Theorem \ref{theo2_urban_floor} reveals that the system performance is bottlenecked by the BS-to-UAV link. Increasing UAV power $P_2$ yields diminishing returns.
\end{remark}

\section{Energy Efficiency Optimization}
This section develops a framework to maximize the system EE, defined as successfully delivered bits per unit of energy (bits/Joule). We formulate a realistic EE model that captures the core trade-off: while increasing $N$ enhances diversity, it also incurs non-trivial time and energy overheads for port selection \cite{wong2021fluid}.

With $T_{\text{block}} = L/W_{\text{band}}$ as the total duration and $T_{\text{sw}}(N) = N\tau_p$ as the FAS selection overhead, the effective data transmission time is $T_{\text{tx}}(N) = T_{\text{block}} - T_{\text{sw}}(N)$. With transmit power $P_2$, static circuit power $P_c$, and FAS switching power $P_{\text{sw}}$, the total energy is $E_{\text{total}} = P_2 T_{\text{tx}}(N) + P_c T_{\text{block}} + P_{\text{sw}} T_{\text{sw}}(N)$. The EE is therefore
\begin{equation}\label{eq:ee_model}
 \text{EE} = \frac{B(1 - \bar{\epsilon}_O^{\mathrm{US}})}{P_2(T_{\text{block}} \!-\! N\tau_p)\! + \!P_c T_{\text{block}} + P_{\text{sw}}N\tau_p}.
\end{equation}
Our objective is to maximize the EE by jointly optimizing $L, Z_{\rm U}, P_2$, and $N$ as
\begin{subequations}
\label{eq:p1}
\begin{align}
\mathcal{P}1: \max_{L, Z_U, P_2, N} \quad & \text{EE}(L, Z_U, P_2, N) \\
\text{s.t.} \quad & \overline{\epsilon}_O^{\mathrm{US}}(L, Z_U, P_2, N) \le \epsilon_{\text{th}}, \label{eq:p1_c1} \\
& 0 < P_2 \le P_{\text{max}}, \label{eq:p1_c2} \\
& Z_{\text{min}} \le Z_U \le Z_{\text{max}}, \label{eq:p1_c3} \\
& L_{\text{min}} \le L \le L_{\text{max}}, \label{eq:p1_c4} \\
& N_{\text{min}} \le N \le N_{\text{max}}, \label{eq:p1_c5} \\
& N\tau_p < T_{\text{block}}. \label{eq:p1_c6}
\end{align}
\end{subequations}
Constraint \eqref{eq:p1_c1} ensures reliability, while \eqref{eq:p1_c2}--\eqref{eq:p1_c6} are operational bounds.

Problem $\mathcal{P}1$ is a non-convex mixed-integer nonlinear program (MINLP). We solve it using a hierarchical search.

\subsection{Transmit Power Minimization}
For a fixed $(\hat{L}, \hat{Z}_U, \hat{N})$, the minimum $P_2^*$ satisfying \eqref{eq:p1_c1} is found. This is a simple 1-D optimization given by
\begin{subequations}
\begin{align}
\mathcal{P}1.1: \quad \min_{P_2} \quad & P_2 \\
\mathrm{s.t.} \quad & \overline{\epsilon}_{O}(P_2; \hat{L}, \hat{Z}_U, \hat{N}) \le \epsilon_{\rm th}, \quad P_2 \le P_{\max}.
\end{align}
\end{subequations}
This is efficiently solved via a bisection search, leveraging the monotonicity of $\overline{\epsilon}_{O}$ in $P_2$.

\subsection{Optimal Port Number Determination}
For a fixed $(\hat{L}, \hat{Z}_U)$, we find the optimal $N^*$ that maximizes EE, given by
\begin{subequations}
\begin{align}
\mathcal{P}1.2: \max_{N} \quad & \text{EE}(N, P_2^*(N); \hat{L}, \hat{Z}_U) \\
\mathrm{s.t.} \quad & N \in [N_{\min}, N_{\max}], \quad N\tau_p < \hat{L}/W_{\text{band}},
\end{align}
\end{subequations}
where $P_2^*(N)$ is the solution from $\mathcal{P}1.1$. This is solved by a 1-D search over the integer $N$.

\subsection{Optimal Flying Height Determination}
For a fixed $\hat{L}$, the optimal altitude $Z_U^*$ is found by
\begin{subequations}
\begin{align}
\mathcal{P}1.3: \quad \max_{Z_U} \quad & \text{EE}^*(Z_U; \hat{L}) \\
\mathrm{s.t.} \quad & Z_U \in [Z_{\min}, Z_{\max}],
\end{align}
\end{subequations}
where $\text{EE}^*(Z_U; \hat{L})$ is the maximum EE obtained from solving $\mathcal{P}1.2$. This is also solved via a 1-D search.

The final solution is found by performing an exhaustive search over the discrete blocklength $L$ and solving $\mathcal{P}1.3$ for each $L$.

\begin{table}[t]
\caption{Simulation Parameters for Urban Scenario}
\label{tab:sim_params_urban}
\centering
\scriptsize
\setlength{\tabcolsep}{4pt} 
\renewcommand{\arraystretch}{1.12}

\resizebox{\columnwidth}{!}{%
\begin{tabular}{|l|c|c|}
\hline
\textbf{Parameter} & \textbf{Symbol} & \textbf{Urban Scenario} \\ \hline
Data bits per packet & $B$ & $80\ \text{bits}$ \\ \hline
Carrier frequency & $f_c$ & $2.5\ \text{GHz}$ \\ \hline
Noise power & $\sigma^{2}$ & $-100\ \text{dBm}$ \\ \hline
System bandwidth & $W_{\text{band}}$ & $10\ \text{MHz}$ \\ \hline
UAV flying radius & $r$ & $50\ \text{m}$ \\ \hline
Static circuit power &$P_c$ & $5\ \text{dBm}$ \\ \hline
FAS switching power & $P_{\rm sw}$ & $0\ \text{dBm}$ \\ \hline
Port processing time & $\tau_p$ & $2~\mu$s \\ \hline 
Aperture of length (default) & $W$ & $0.5\lambda$\\ \hline
BLER threshold & $\epsilon_{\rm th}$ & $10^{-3}$\\ \hline
BS coordinates & $(X_B, Y_B, Z_B)$ & $(100, 0, 40)\ \text{m}$ \\ \hline
UE coordinates & $(X_E, Y_E, Z_E)$ & $(-100, 100, 0)\ \text{m}$ \\ \hline
LoS probability constants & $(a,b)$ & $12.08,\;0.11$ \\ \hline
Additional LoS path loss & $\eta_{\text{LoS}}$ & $1.6\ \text{dB}$ \\ \hline
Additional NLoS path loss & $\eta_{\text{NLoS}}$ & $23\ \text{dB}$ \\ \hline
LoS Nakagami-$m$ factor & $m_{\text{LoS}}$ & $5$ \\ \hline
NLoS Nakagami-$m$ factor & $m_{\text{NLoS}}$ & $1$ (Rayleigh) \\ \hline
\end{tabular}%
}
\end{table}

\section{Numerical Results and Discussion}
In this section, we present numerical results to validate our analytical framework, quantify the performance gains of the FAS, and investigate the proposed EE optimization. Unless specified otherwise, the simulation parameters for the urban scenario are listed in Table \ref{tab:sim_params_urban}.

\begin{figure}[t]
\centering
\subfloat[Validation of the analytical framework for the urban scenario. Parameters: $L=100, N=2, W=0.5\lambda, m_{\text{LoS}}=5, m_{\text{NLoS}}=1, P_1=40$ dBm.]{%
    \includegraphics[width=0.45\columnwidth]{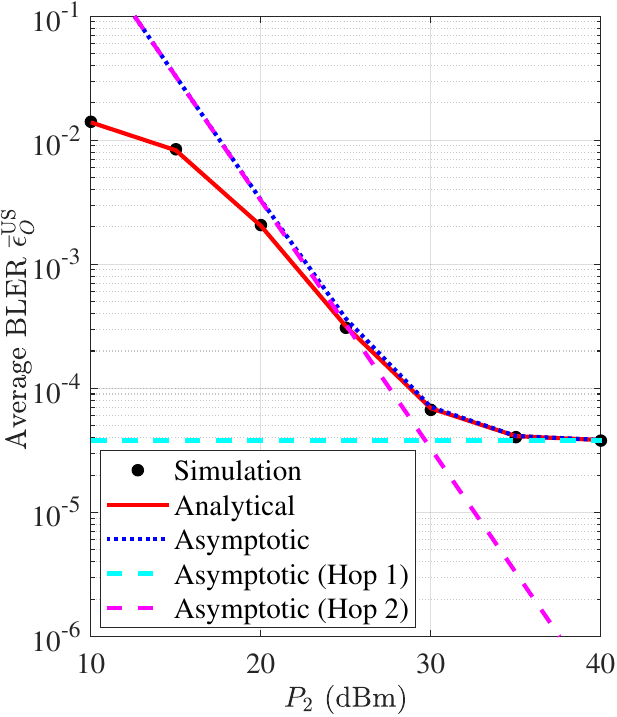}
    \label{fig:validation_urban}
}
\hfil
\subfloat[Performance comparison between FAS ($N=2$) and FPA ($N=1$). Parameters are identical to those in (a).]{%
    \includegraphics[width=0.45\columnwidth]{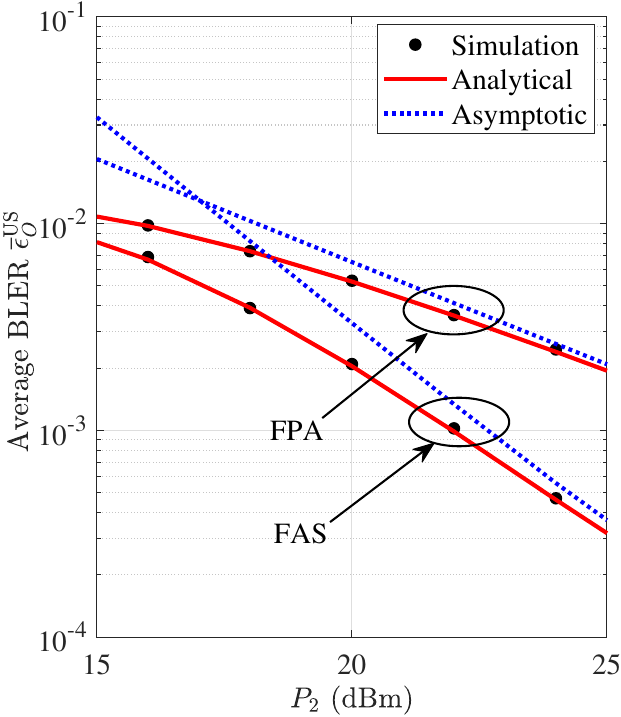}
    \label{fig:fas_vs_fpa_urban}
}
\caption{BLER performance: (a) Analytical validation and (b) Baseline comparison.}
\label{fig:validation_and_fpa}
\vspace{-3mm} 
\end{figure}

\begin{figure}[t]
\centering
\subfloat[End-to-end BLER versus FAS aperture size $W$ for a fixed $N=8$. Parameters: $L=100, m_{\text{LoS}}=5, m_{\text{NLoS}}=1, P_1=40$ dBm.]{%
    \includegraphics[width=0.45\columnwidth]{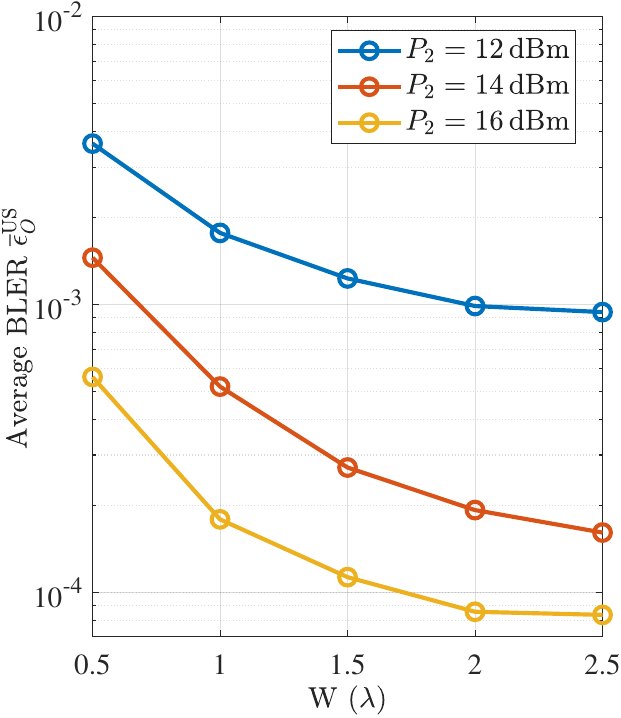}
    \label{fig:bler_vs_w_urban}
}
\hfil
\subfloat[Minimum UAV transmit power $P_2^*$ versus UAV altitude $Z_U$ for $\epsilon_{\text{th}}=10^{-3}$, $L=200$, and $P_1=46$~dBm.]{%
    \includegraphics[width=0.45\columnwidth]{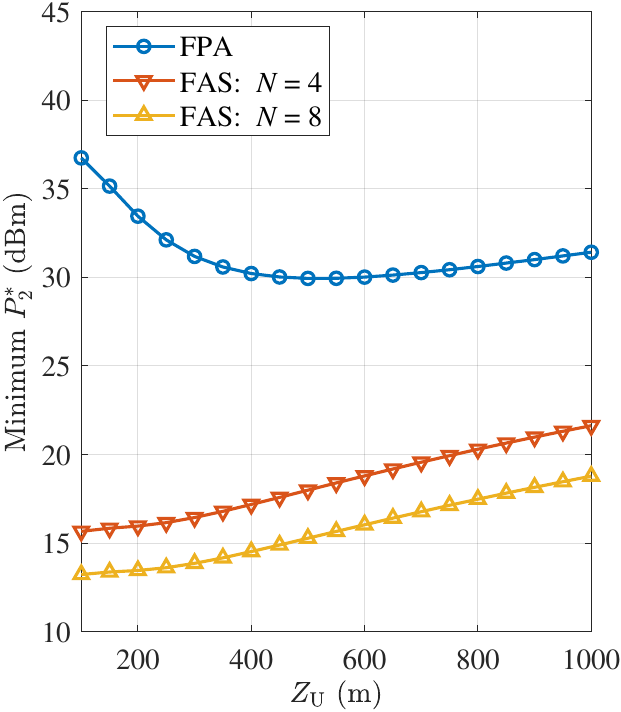}
    \label{fig:power_vs_alt_urban}
}
\caption{BLER vs. aperture size and UAV power vs. altitude.}
\label{fig:bler_and_power}
\vspace{-3mm} 
\end{figure}

\begin{figure}[t]
\centering
\subfloat[EE versus the number of FAS ports $N$ for different blocklengths.]{%
    \includegraphics[width=0.45\columnwidth]{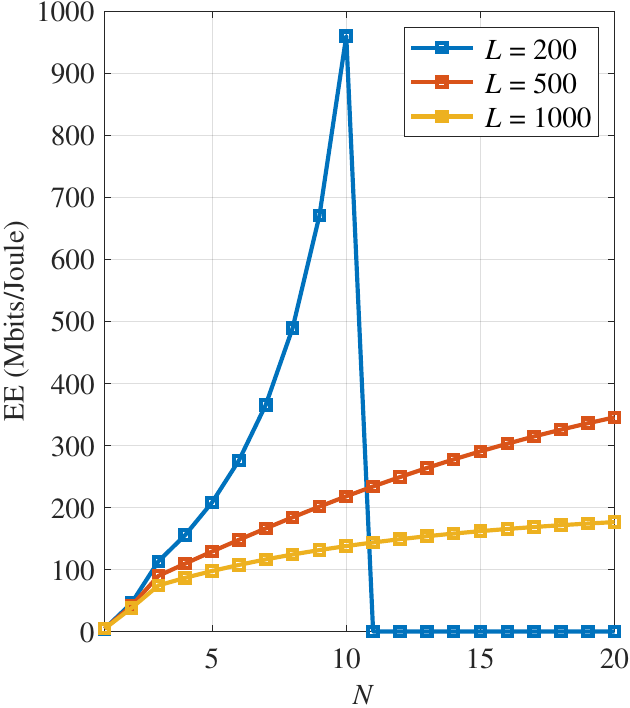}
    \label{fig:ee_vs_n_urban}
}
\hfil
\subfloat[Maximum achievable EE contours over $(Z_U, L)$, highlighting the global optimum (red star).]{%
    \includegraphics[width=0.49\columnwidth]{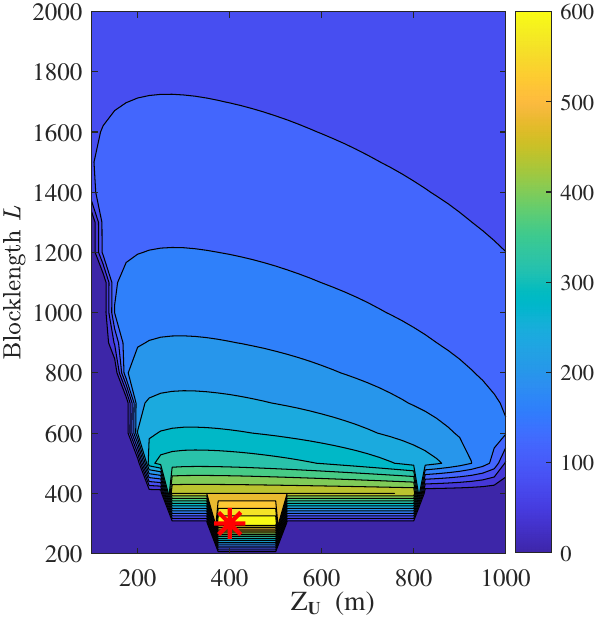}
    \label{fig:ee_contour_urban}
}
\caption{EE analysis and global optimization.}
\label{fig:ee_analysis}
\vspace{-3mm} 
\end{figure}

Fig. \ref{fig:validation_urban} validates our analytical framework. The analytical curve perfectly matches the Monte Carlo simulation points, confirming the accuracy of our derived BLER expressions for the complex probabilistic urban scenario. The high-SNR asymptotic analysis also correctly captures the diversity slope. Notably, an error floor appears as the system becomes bottlenecked by the fixed first-hop link, which aligns with the theoretical behavior of a DF relay.

Fig. \ref{fig:fas_vs_fpa_urban} quantifies the performance advantage of FAS. Even a minimal 2-port FAS provides a substantial power gain over the conventional FPA scheme. The steeper slope of the FAS curve confirms its higher diversity order, demonstrating that FAS can combat fading more effectively in the challenging urban environment. For instance, to achieve a BLER of $10^{-4}$, FAS requires approximately 3~dB less power than FPA.

Fig. \ref{fig:bler_vs_w_urban} investigates the impact of the FAS aperture size $W$ for a fixed $N=8$. The BLER decreases as $W$ increases because a larger aperture reduces spatial correlation, enhancing the effective diversity gain. However, this improvement exhibits diminishing returns. As $W$ becomes sufficiently large (e.g., $W > 2\lambda$), the channels decorrelate, and further increasing the aperture yields only marginal gains. This saturation effect reveals a key design insight: an aperture must be large enough for diversity, but an excessively large footprint is unnecessary.

Fig. \ref{fig:power_vs_alt_urban} investigates the minimum UAV transmit power $P_2^*$ required to meet the target BLER as a function of altitude $Z_U$. FAS substantially outperforms the FPA, achieving over 15~dB of power savings at the optimal altitude. Critically, the urban scenario exhibits a distinct convex trend, revealing an optimal altitude at $Z_U^* \approx 450$~m. This demonstrates the fundamental trade-off between avoiding NLoS blockage (favoring high altitudes) and minimizing path loss (favoring low altitudes).

Fig. \ref{fig:ee_vs_n_urban} presents the central finding of this paper. It plots the maximized EE as a function of $N$, revealing a clear quasi-concave relationship. Initially, increasing $N$ enhances diversity, which significantly reduces the required $P_2^*$ and boosts EE. However, beyond an optimal number of ports, the linear increase in time and energy overhead from port selection begins to outweigh the diminishing diversity gains, causing the EE to decrease. This confirms the critical trade-off captured by our realistic model. Furthermore, for the low-latency case of $L=200$, the EE abruptly drops to zero at $N=10$ due to the causality constraint $N\tau_p < L/W_{\text{band}}$.

To visualize the global optimization landscape, Fig. \ref{fig:ee_contour_urban} depicts the maximum achievable EE as a function of $Z_U$ and $L$, where $N$ is optimized pointwise. The contour map reveals that the optimal operating point, denoted by a star, corresponds to the minimum blocklength $L^* = 300$ and an altitude of $Z_U^* \approx 400$~m. This result corroborates the proposed framework which jointly addresses the interplay among blockage probability, path loss, and FAS overheads.

\section{Conclusion}
This paper established a foundational framework for analyzing UAV-enabled short-packet communication by leveraging FAS-assisted UAV relaying. In the short-packet regime, we derived novel closed-form BLER expressions for probabilistic NLoS urban environments. Our high-SNR asymptotic analysis quantified the fundamental diversity order and identified the first-hop link as the performance bottleneck. Building on this analysis, we optimized system EE by formulating a novel problem that incorporates the practical time and energy overheads of FAS port selection. This revealed a key trade-off between diversity gain and operational cost, leading to two key insights: the existence of an optimal, finite number of FAS ports, and an optimal UAV altitude governed by the balance between blockage and path loss. This work provides valuable design insights for reliable, energy-efficient UAV systems and can be extended to multi-user scenarios or the impact of imperfect CSI.

\begin{appendices}
\section{Proof of Lemma \ref{lema1_urban}}\label{plema1_urban_proof}
The power gain $Y = |g|^2$ of a Nakagami-$m$ channel (with $m_1$ and $\Omega=1$) follows a Gamma distribution. Its probability density function (PDF) is $f_Y(y) = m_1^{m_1} y^{m_1-1} e^{-m_1 y} / \Gamma(m_1)$. The cumulative distribution function (CDF) is found by integration as
\begin{equation}
F_Y(y) = \int_0^y f_Y(t) dt = \frac{1}{\Gamma(m_1)} \int_0^{m_1 y} u^{m_1-1} e^{-u} du,
\end{equation}
which, by definition, is $F_Y(y) = \gamma(m_1, m_1 y) / \Gamma(m_1)$, where $\gamma(\cdot,\cdot)$ is the lower incomplete gamma function. For integer $m_1$, this is equivalent to $F_Y(y) = 1 - e^{-m_1 y} \sum_{k=0}^{m_1-1} \frac{(m_1 y)^k}{k!}$.

The instantaneous SNR for link type $k$ is $\gamma_1^k = \bar{\gamma}_1^k(\theta) Y$. The CDF of $\gamma_1^k$ is found by scaling, where
\begin{align}
F_{\gamma_1^k}(x) &= \Pr(\bar{\gamma}_1^k(\theta) Y \le x) = F_Y\left(\frac{x}{\bar{\gamma}_1^k(\theta)}\right) \nonumber \\
&= 1 - e^{-\frac{m_1 x}{\bar{\gamma}_1^k(\theta)}} \sum\nolimits_{k=0}^{m_1-1} \frac{1}{k!}\left(\frac{m_1 x}{\bar{\gamma}_1^k(\theta)}\right)^k.
\end{align}
Substituting the definition $\vartheta_1^k(\theta) = m_1 / \bar{\gamma}_1^k(\theta)$ completes the proof.

\section{Proof of Lemma \ref{lema2_urban}}\label{plema2_urban_proof}
Let $Y_{\max} \triangleq |h_{\rm FAS}|^2$ be the selected channel power gain from \eqref{eq:analytical_model_def}, and let $Y_n \triangleq \lambda_n|g_n|^2$. Since the base components $|g_n|^2$ are i.i.d., the $Y_n$ are independent. The CDF of the maximum is the product of the individual CDFs as
\begin{equation}
F_{Y_{\max}}(y) = \Pr(\max_n Y_n \le y) = \prod_{n=1}^{N_{\rm eff}} F_{Y_n}(y).
\end{equation}
From Appendix \ref{plema1_urban_proof}, the CDF of a single Nakagami-$m$ power gain $|g_n|^2$ (with parameter $m_2^k$) is $F_{|g_n|^2}(y) = \gamma(m_2^k, m_2^k y) / \Gamma(m_2^k)$. By scaling, the CDF of $Y_n = \lambda_n |g_n|^2$ is
\begin{equation}
F_{Y_n}(y) = F_{|g_n|^2}(y/\lambda_n) = \frac{\gamma(m_2^k, m_2^k y / \lambda_n)}{\Gamma(m_2^k)}.
\end{equation}
Substituting this back into the product gives
\begin{equation}
F_{|h_{\rm FAS}|^2}(y) = \prod_{n=1}^{N_{\rm eff}} \frac{\gamma\left(m_2^k, \frac{m_2^k y}{\lambda_n}\right)}{\Gamma(m_2^k)}.
\end{equation}
The instantaneous SNR is $\gamma_2^k = \overline{\gamma}_2^k(\theta) |h_{\rm FAS}|^2 / (\sum_{n=1}^{N_{\rm eff}}\lambda_n)$. Its CDF is found by scaling $F_{|h_{\rm FAS}|^2}(y)$ with $y = x (\sum_{n=1}^{N_{\rm eff}}\lambda_n) / \overline{\gamma}_2^k(\theta)$. Using the definition $\vartheta_2^k(\theta) = m_2^k \sum_{n=1}^{N_{\rm eff}}\lambda_n / \overline{\gamma}_2^k(\theta)$, the argument of the gamma function becomes
\begin{equation}
\frac{m_2^k y}{\lambda_n} = \frac{m_2^k}{\lambda_n} \left( \frac{x \sum_{n=1}^{N_{\rm eff}}\lambda_n}{\overline{\gamma}_2^k(\theta)} \right) = \frac{x}{\lambda_n} \left( \frac{m_2^k \sum_{n=1}^{N_{\rm eff}}\lambda_n}{\overline{\gamma}_2^k(\theta)} \right) = \frac{x \vartheta_2^k}{\lambda_n}.
\end{equation}
Thus, $F_{\gamma_2^k}(x) = \prod_{n=1}^{N_{\rm eff}} \gamma(m_2^k, x \vartheta_2^k / \lambda_n) / \Gamma(m_2^k)$. For integer $m_2^k$, this is identical to the expression in Lemma \ref{lema2_urban}.

\section{Proof of Theorem \ref{theo1_urban}}\label{proof_theo1_urban}
In the high-SNR regime ($x \to 0$), the argument $z = x \vartheta_2^k / \lambda_n \to 0$. The CDF of a single branch, $F_{Y_n}(z) = \gamma(m_2^k, z) / \Gamma(m_2^k)$, can be approximated by its first series term as
\begin{equation}
\frac{\gamma(s, z)}{\Gamma(s)} = \frac{1}{\Gamma(s)} \int_0^z t^{s-1} e^{-t} dt \approx \frac{1}{\Gamma(s)} \frac{z^s}{s} = \frac{z^s}{\Gamma(s+1)}.
\end{equation}
Applying this to the product CDF $F_{\gamma_2^k}(x)$, which is
\begin{equation}
F_{\gamma_2^k}(x) = \prod_{n=1}^{N_{\rm eff}} \frac{\gamma\left(m_2^k, \frac{x \vartheta_2^k}{\lambda_n}\right)}{\Gamma(m_2^k)}. 
\end{equation}
Resort to the approximation $\gamma(m_2^k, z) \approx z^{m_2^k} / m_2^k$, we have
\begin{align}
F_{\gamma_2^k}(x) &\approx \prod_{n=1}^{N_{\rm eff}} \frac{1}{\Gamma(m_2^k+1)} \left(\frac{x \vartheta_2^k}{\lambda_n}\right)^{m_2^k} \nonumber \\
&= \underbrace{\left( \frac{(\vartheta_2^k)^{m_2^k}}{\Gamma(m_2^k+1)} \right)^{N_{\rm eff}} \left( \prod_{n=1}^{N_{\rm eff}} \lambda_n^{-m_2^k} \right)}_{\mathcal{C}_k} x^{m_2^k N_{\rm eff}}.
\end{align}
Let $F_{\gamma_2^k}(x) \approx \mathcal{C}_k x^{m_2^k N_{\rm eff}}$, where $\mathcal{C}_k$ is the constant defined above. We substitute this approximation into the average BLER integral \eqref{epaprir_urban}, yielding
\begin{align}
\bar{\epsilon}^k_2 &\approx \chi \int_{\rho_L}^{\rho_H} \mathcal{C}_k x^{m_2^k N_{\rm eff}} dx
= \chi \mathcal{C}_k \left[ \frac{x^{m_2^k N_{\rm eff} + 1}}{m_2^k N_{\rm eff} + 1} \right]_{\rho_L}^{\rho_H} \nonumber \\
&= \frac{\chi \mathcal{C}_k}{m_2^k N_{\rm eff} + 1} \left( \rho_H^{m_2^k N_{\rm eff} + 1} - \rho_L^{m_2^k N_{\rm eff} + 1} \right).
\end{align}
Substituting the full expression for $\mathcal{C}_k$ back gives the result in Theorem \ref{theo1_urban}. This completes the proof.
\end{appendices}


\begin{thebibliography}{99}

\bibitem{saad2019vision}
W.~Saad, M.~Bennis, and M.~Chen, ``A vision of {6G} wireless systems:
 Applications, trends, technologies, and challenges,'' \emph{IEEE Netw.},
 vol.~34, no.~3, pp. 134--142, May 2020.
 
\bibitem{she2017radio}
C.~She, C.~Yang, and T.~Q.~S. Quek, ``Radio resource management for
 ultra-reliable and low-latency communications,'' \emph{IEEE Commun. Mag.},
 vol.~55, no.~6, pp. 72--78, Jun. 2017.

\bibitem{polyanskiy2010channel}
Y.~Polyanskiy, H.~V. Poor, and S.~Verdú, ``Channel coding rate in the finite
 blocklength regime,'' \emph{IEEE Trans. Inf. Theory}, vol.~56, no.~5, pp.
 2307--2359, May 2010.

\bibitem{wu2021overview}
Y.~Zeng, Q.~Wu, and R.~Zhang, ``Accessing from the sky: A tutorial on {UAV}
 communications for {5G} and beyond,'' \emph{Proc. IEEE}, vol. 107, no.~12,
 pp. 2327--2375, Dec. 2019.

\bibitem{yuan2022per}
L. Yuan, N. Yang, F. Fang, and Z. Ding, ``Performance analysis of UAV-assisted short-packet cooperative communications," \emph{IEEE Trans. Veh. Technol.}, vol. 71, no. 4, pp. 4471-4476, Apr. 2022.

\bibitem{new2025a}
X. Zhu, F. R. Ghadi, T. Wu, K. Meng, C. Wang, and G. Zhou, ``On the Fundamental Scaling Laws of Fluid Antenna Systems,"
\emph{arXiv preprint}, arXiv:2511.03415, 2025.

\bibitem{wong2021fluid}
K. K. Wong, A.~Shojaeifard, K.-F. Tong, and Y.~Zhang, ``Fluid antenna systems,''
 \emph{IEEE Trans. Wireless Commun.}, vol.~20, no.~3, pp. 1950--1968, Mar.
 2021.

\bibitem{wong2020pel}
X. Zhu, K. K. Wong, B. Tang, W. Chen, and C. B. Chae, ``Fluid reconfigurable intelligent surface (FRIS) enabling secure wireless communications,"
\emph{arXiv preprint}, arXiv:2511.15860, 2025.

\bibitem{wong2022fama}
K. K. Wong and K.-F. Tong, ``Fluid antenna multiple access," \emph{IEEE Trans. Wireless Commun.}, vol. 21, no. 7, pp. 4801-4815, Jul. 2022.

\bibitem{wong2024compact}
K. K. Wong, C.-B. Chae, and K. F. Tong, ``Compact ultra massive antenna array: A simple open-loop massive connectivity scheme,” \emph{IEEE Trans. Wireless Commun.}, vol. 23, no. 6, pp. 6279-6294, Jun. 2024.

\bibitem{cha2025uav} 
X. Zhu, K. K. Wong, Q. Wu, H. Shin, and Y. Zhang, ``Fluid Antenna System-Enabled UAV-to-Ground Communications,"
\emph{arXiv preprint}, arXiv:2511.17416, 2025.

\bibitem{li2025radi}
Z. Li, Z. Gao, B. Ning, and Z. Wang, ``Radiation pattern reconfigurable FAS-empowered interference-resilient UAV communication," \emph{IEEE J. Sel. Areas Commun.}, early access, doi: 10.1109/JSAC.2025.3617928.

\bibitem{xu2025fluid}
X. Xu, H. Xu, H. Yu, Y. Liu, and M. Chen, ``Fluid antenna system (FAS)-assisted 3D UAV positioning performance optimization," in \emph{Proc. IEEE Int. Conf. Commun.}, Montreal, QC, Canada, 2025, pp. 2260-2265.


\bibitem{zhu2025fas}
X. Zhu, et al., ``Fluid antenna systems: A geometric approach to error probability and fundamental limits," \emph{arXiv preprint}, arXiv:2509.08815, 2025.

\bibitem{hour2014opt}
A. Al-Hourani, S. Kandeepan, and S. Lardner, ``Optimal LAP altitude for maximum coverage," \emph{IEEE Wireless Commun. Lett.}, vol. 3, no. 6, pp. 569-572, Dec. 2014.



\end{thebibliography}
\end{document}